\documentclass[12pt]{amsart} 
\usepackage{graphicx}
\usepackage[english]{babel} 
\usepackage{amsmath,amssymb,amsthm,amsfonts} 
\usepackage{latexsym}
\usepackage[all]{xy}
\usepackage{url}
\usepackage{mathrsfs}
\usepackage[latin1]{inputenc}

\usepackage{hyperref}
\hypersetup{
 pdftitle={The Rabin cryptosystem in number fields},
 pdfauthor={Alessandro Cobbe, Andreas Nickel, Akay Schuster}}

\newcommand{\Q}{\mathbb Q}
\newcommand{\N}{\mathbb N}
\newcommand{\Z}{\mathbb Z}
\newcommand{\C}{\mathbb C}

\newcommand{\F}{\mathbb F}
\newcommand{\p}{\mathfrak p}

\renewcommand{\O}{\mathcal{O}}
\renewcommand{\p}{\mathfrak{p}}
\newcommand{\q}{\mathfrak{q}}
\renewcommand{\epsilon}{\varepsilon}

\newtheorem{thm}{Theorem}[section]

\newtheorem{prop}[thm]{Proposition}
\newtheorem{algo}[thm]{Algorithm}

\theoremstyle{remark}

\newtheorem{remark}[thm]{Remark}
\newtheorem{example}[thm]{Example}

\theoremstyle{definition}

\title{The Rabin cryptosystem over number fields}
\author{Alessandro Cobbe}
\address{
	Universit\"at der Bundeswehr M\"unchen\\
	Fakult\"at f\"ur Informatik\\
	Werner-Heisenberg Weg 39\\
	85579 Neubiberg\\
	Germany}
\email{alessandro.cobbe@unibw.de}
\author{Andreas Nickel} 
\address{
	Universit\"at der Bundeswehr M\"unchen\\
	Fakult\"at f\"ur Informatik\\
	Werner-Heisenberg Weg 39\\
	85579 Neubiberg\\
	Germany}
\email{andreas.nickel@unibw.de}
\urladdr{https://www.unibw.de/timor/mitarbeiter/univ-prof-dr-andreas-nickel}
\author{Akay Schuster}
\address{
	Universit\"at der Bundeswehr M\"unchen\\
	Fakult\"at f\"ur Informatik\\
	Werner-Heisenberg Weg 39\\
	85579 Neubiberg\\
	Germany}
\email{akay.schuster@unibw.de}

\date{Version of \today}
\subjclass[2020]{94A60, 11T71, 11Y40}
\keywords{Rabin encryption, public-key cryptography, number fields}

\begin{document}
\maketitle
\begin{abstract}
We extend Rabin's cryptosystem to general number fields. We show that recovery of a random plaintext from its ciphertext
is as hard as the integer factorisation problem, provided the modulus
in our scheme has been chosen carefully. We investigate the performance of our new cryptosystem in comparison
with the classical Rabin scheme and a more recent version over the Gaussian integers.
\end{abstract}

\section*{Introduction}
Rabin's cryptosystem \cite{Rabin} was the first asymmetric public-key cryptosystem for which it
was shown that any algorithm which finds one of the possible plaintexts for every Rabin-encrypted ciphertext can be used to factor the modulus. It is performed in the quotient ring $\Z/ n\Z$ where $n = pq$ is the product of two large
distinct prime numbers. The public key is $n$ and the private key is the pair $(p,q)$.
Encryption is very simple: To encrypt a message $m \in \Z/ n\Z$ one just computes
its square; so $c = m^2$ is the ciphertext. For decryption, one uses the knowledge of $p$ and $q$
to compute the square roots of $c \pmod p$ and $c \pmod q$. For this one can use the 
algorithm of Tonelli-Shanks (see \cite[\S 1.5.1]{Cohen}, for instance), but there is an easier method for primes
which are congruent to $3$ modulo $4$: If $y$ is a quadratic residue modulo such a $p$, then
$\pm y^{(p+1)/4}$ are the roots of $y$. Finally, one uses the Chinese Remainder
Theorem to obtain the (in general four) square roots of $c$ modulo $n$. It remains to decide which of these
roots is the original message $m$.

In this paper, we generalise Rabin's scheme to rings of integers in number fields. So let us fix a number field
$K$ with ring of integers $\mathcal{O}_K$ and choose two distinct non-zero prime ideals $\p$ and $\q$ in $\mathcal{O}_K$.
Our scheme is performed in the quotient ring $\mathcal{O}_K / \mathfrak n$ where $\mathfrak{n} = \p \q$.
The public key is the modulus $\mathfrak{n}$ and the private key is the pair $(\p,\q)$.
The ciphertext of a message $m \in \mathcal{O}_K/\mathfrak{n}$ is again $c = m^2$. For decryption one has to compute
the square roots of $c$ modulo $\p$ and  modulo $\q$. For this, we use a generalisation of the Tonelli-Shanks algorithm.
Finally, one uses the Chinese Remainder Theorem to obtain the (again up to four) square roots of $c$ modulo 
$\mathfrak{n}$.

The special case of Gaussian integers
has been treated by Awad, El-Kassar and Kadri \cite{RabinGaussianIntegers}. They use the (extended) Euclidean algorithm
in $\Z[i]$ to make the isomorphism in the Chinese Remainder Theorem explicit. This does not generalise to
arbitrary number fields. We provide a version that only relies on the Euclidean algorithm in $\Z$.

In order to speed up decryption, we also provide an easy and fast way to compute square roots modulo
primes $\p$ in $\mathcal{O}_K$ under certain restrictions on $\p$, thereby generalising the above method
for primes $p \equiv 3 \pmod 4$. This approach works particularly well in number fields of degree $3$,
but does not work in any quadratic number field. This is one reason why we need to allow more freedom in choosing
the field $K$. Indeed, our runtime analysis shows that decryption of ciphertexts of comparable size
becomes a lot faster over carefully chosen cubic than over quadratic fields.

It turns out that it is a subtle question how to represent the public key $\mathfrak{n}$. There are two reasonable
ways to do this: either as a list of generators or in Hermite normal form. We show that for most choices of
$\p$ and $\q$ neither representation is secure. Either one allows us to factor $\mathfrak{n}$. Our considerations led us
to the conclusion that the best choices are prime ideals of the form $\p= (p)$, where $p \in \Z$ is prime,
i.e.\ we consider rational primes $p$ which are inert in $K$. So we also need conditions that guarantee us that
a given rational prime remains prime in $K$. We will provide a class of number fields, where this can be
checked by a simple congruence condition. So the public key in our scheme is just $n = pq$
where $p,q \in \Z$ are primes that are inert in $K$. Then, as in Rabin's original scheme, an algorithm
which finds one of the possible plaintexts for each ciphertext allows $n$ to be factored.

We also address the question how to find the original message $m$ among the usually four square roots
of the ciphertext $c=m^2$. In the classical case, Williams \cite{Williams} proposed a method
whenever both primes $p$ and $q$ are congruent to $3$ modulo $4$. It suffices to add two extra bits
to the ciphertext: the parity and the Jacobi symbol $\left(\frac{m}{n}\right)$ of the message $m$.
If we assume in addition that $p$ and $q$ are inert in $K$, we will show that it is still possible 
to uniquely identify the original message by adding two extra bits.

In the following table we summarise the main aspects of the Rabin system and some of its modifications and compare them to our results. Note in particular that the proof that the security of the Rabin system is equivalent to the factorisation problem does not work anymore if the disambiguation of the possible square roots is attained via some redundancy in the plaintext, as is the case in \cite{RabinGaussianIntegers}. 
The way we use the techniques from \cite{EliaPivaSchipani} to eliminate the ambiguity in the decryption, as well as other computational aspects will be explained in the paper.

\begin{table}[htbp]
	\centering
	\caption{Comparison of variants of the Rabin cryptosystem}
	\begin{tabular}{p{2.3cm}p{1.4cm}p{3.1cm}p{3.5cm}p{3cm}}
		\textbf{Reference} &
		\textbf{Ring} &
		\textbf{Condition on \newline primes} &
		\textbf{Root selection} &
		\textbf{Equivalent to \newline  factoring $n$} \\
		\hline
		\cite{Rabin}
		& $\mathbb Z$
		& arbitrary $p,q$
		& none; four possible \newline square roots
		& yes
		\\ \hline
		
		\cite{Williams}
		& $\mathbb Z$
		& $p\equiv3 \pmod 8,\newline q\equiv7\pmod 8$
		& parity and \newline Jacobi symbol
		& yes
		\\ \hline
		
		\cite{EliaPivaSchipani}
		& $\mathbb Z$
		& arbitrary $p,q$
		& various techniques
		& yes
		\\ \hline
		
		\cite{RabinGaussianIntegers}
		& $\mathbb Z[i]$
		& $p,q\equiv3\pmod4$
		& redundancy
		& no
		\\ \hline
		
		This paper
		& $\mathcal O_K$
		& $p,q$ inert in $K$
		& as in \cite{EliaPivaSchipani}
		& yes
		\\
	\end{tabular}
\end{table}

Finally, we note that Petukhova and Tronin \cite{RSA-Dedekind} have generalised the RSA scheme
to general Dedekind domains with finite residue fields. However, they have not addressed the following
questions: (i) how to find such Dedekind domains, (ii) how to make the required computations explicit,
(iii) when is the scheme secure? Though we focus on Rabin's cryptosystem in this paper, many of our
considerations also apply to this generalisation of the RSA scheme. (i) Rings of integers in number fields
are a natural source of Dedekind domains with finite residue fields. (ii) Our remarks on computations
in residue rings of the form $\mathcal{O}_K / \mathfrak{n}$ also apply to the generalised RSA scheme.
(iii) The public key in the generalised RSA scheme is a pair $(\mathfrak{n},e)$, where $e$ is
coprime with the cardinality of $\left(\mathcal{O}_K/\mathfrak{n}\right)^{\times}$.
Since we show that the knowledge of $\mathfrak{n}$ (if given as a list of generators or in Hermite normal form)
often suffices to factor it, one needs to impose the same conditions on $\p$ and $\q$
also in the generalised RSA scheme.

\subsection*{Acknowledgments}
The authors would like to thank the anonymous referees for the valuable suggestions which helped to improve the quality of the paper.
Moreover, we acknowledge financial support by Universität der Bundeswehr München.

\section{The public key}
\subsection{Number fields and rings of integers}
Let us first explain how we represent number fields.
Let $K$ be a number field of degree $d$ over $\Q$.
Let $\theta \in K$ be a primitive element so that $K = \Q(\theta)$, which means that
$1 = \theta^0, \theta, \theta^2, \dots, \theta^{d-1}$ constitute a $\Q$\nobreakdash-basis of $K$.
The minimal polynomial of $\theta$ is the unique monic polynomial 
${g_{\theta} \in \Q[x]}$ of lowest degree with root $\theta$. It is irreducible of degree $d$ and there is 
a field isomorphism
\begin{equation} \label{eqn:K-iso}
	\Q[x] / (g_{\theta}) \xrightarrow{\, \simeq \, } K = \Q(\theta), \,
	x \mapsto \theta.
\end{equation}
So if we `choose' a number field, we really pick a monic irreducible polynomial ${g \in \Q[x]}$
and consider the field $\Q[x]/(g)$. If we want to view this field as a subfield of the complex numbers,
we choose a root $\theta \in \C$ of $g$ and apply \eqref{eqn:K-iso}. Then one has $g_{\theta} = g$
and the set of embeddings of $\Q[x]/(g)$ into $\C$ is in one-to-one correspondence with the roots of $g$.

Let $\O_K$ be the ring of integers in $K$. Then one can always choose $\theta \in \O_K$ so that
$g_{\theta}$ has integral coefficients. The isomorphism \eqref{eqn:K-iso} then induces an isomorphism of rings
\[
	\Z[x] / (g_{\theta}) \simeq \Z[\theta].
\]
The ring $\Z[\theta]$ is always contained in $\O_K$, but we do not have equality in general.
This will be crucial in the following, as we get the ring $\Z[\theta]$ for free once we have
$g_{\theta}$, but it might be rather expensive to compute $\O_K$ in general.
For instance, Zassenhaus's Round 2 Algorithm \cite[\S 6.1.4]{Cohen} first factors the discriminant
of $g_{\theta}$, which is a hard problem when the discriminant is large.
Luckily, it will not be necessary to determine $\O_K$ for our purposes (this will be explained below). 
Alternatively, one can stick to (classes of) number fields, where the ring of integers has been determined.

\begin{example}[Quadratic number fields]
	Let $\delta \in \Z$ be square-free, $\delta \not=0,1$. Then $x^2-\delta \in \Z[x]$ is irreducible and
	$\Q[x]/(x^2-\delta) \simeq \Q(\sqrt{\delta})$.
	The ring of integers in $\Q(\sqrt{\delta})$ is $\Z[\omega]$, where $\omega = \sqrt{\delta}$ if 
	$\delta \equiv 2$ or $3 \pmod 4$, whereas $\omega = (1+\sqrt{\delta})/2$ if ${\delta \equiv 1 \pmod 4}$ 
	(see \cite[Proposition 5.1.1]{Cohen}, for example). So in the latter case the
	index of $\Z[\sqrt{\delta}]$ in $\O_{\Q(\sqrt{\delta})}$ is $2$.
\end{example}

\begin{example}[Cyclotomic fields]
	Let $m$ be a positive integer such that $m \not\equiv 2 \pmod 4$. Let $\zeta_m \in \C$ be a primitive
	$m$-th root of unity. Then $\Q(\zeta_m)$ has degree $\varphi(m)$ over $\Q$, where $\varphi$ denotes
	Euler's totient function. Here one always has $\O_{\Q(\zeta_m)} = \Z[\zeta_m]$ by \cite[Proposition 9.1.2]{Cohen}.
\end{example}

Every $\alpha \in K$ can uniquely be written as $\alpha = \sum_{i=0}^{d-1} \alpha_i \theta^i$ with rational coefficients
$\alpha_0, \dots, \alpha_{d-1}$. If we use this representation for elements in $K$, addition becomes very easy.
For multiplication one essentially has to compute once the representations for $\theta^{d+k}$ with $0 \leq k < d-1$.

%

\subsection{Prime ideals} \label{subsec:prime-ideals}
For our scheme we need to work in residue fields of the form
$\O_K/\p$, where $\p$ is a non-zero prime ideal of $\O_K$. However, we want to avoid computing $\O_K$
and to work with $\Z[\theta]$ instead. 

Let $p$ be the rational prime below $\p$ and $\p_{\theta} := \p \cap \Z[\theta]$ the prime ideal
of $\Z[\theta]$ below $\p$. The index of $\Z[\theta]$ in $\O_K$ is finite.
Whenever it is not divisible by $p$, the natural inclusion
\[
	\Z[\theta]/\p_{\theta} \rightarrow \O_K/\p 
\]
is an isomorphism. As $p$ will be chosen to be a large prime, the probability that $p$ divides the unknown
index $[\O_K: \Z[\theta]]$ is negligibly small. However, there are even methods to guarantee this.

Let us denote the discriminant of a polynomial $g$ by $\Delta(g)$ and the discriminant of $K$ by $\Delta_K$.
If $K = \Q(\theta)$, then we have
\[
	\Delta(g_{\theta}) = \Delta_K \cdot [\O_K:\Z[\theta]]^2.
\]
So it is sufficient to choose a prime $p$ that does not divide $\Delta(g_{\theta})$ and
the latter is computable (see \cite[\S 3.3]{Cohen}). However,  we can also use 
Mahler's bound \cite[Corollary on p.\ 261]{Mahler}
\[
	\left| \Delta(g_{\theta}) \right| < C_{\theta} := d^d L(g_{\theta})^{2d-2},
\]
where we recall that $d = \deg(g_{\theta})$ and $L(g_{\theta}) := \sum_{i=0}^d |a_i|$ if $g = \sum_{i=0}^d a_i x^i$.
This bound is very easy to compute and choosing $p$ larger than $C_{\theta}$ will be sufficient for our purposes.

In order to choose a prime ideal $\p$ we first choose a rational prime $p > C_{\theta}$. Then the polynomial
$g_{\theta} \pmod p \in \F_p[x]$ is square-free and we let
\[
	g_{\theta} \equiv \prod_{j=1}^r g_j(x) \pmod p
\]
be its decomposition into irreducible factors in $\F_p[x]$. Here $g_j(x) \in \Z[x]$ are taken to be monic.
Note that there are efficient algorithms for the factorisation of polynomials modulo $p$, which indeed simplify
if the polynomial is square-free \cite[\S 3.4]{Cohen}. Then by \cite[Theorem 4.8.13]{Cohen} choosing a
prime ideal above $p$ is the same as choosing a factor $g_j$. Indeed the prime ideals are given by
$\p_j = (p, g_j(\theta))$ and the residue degree $[\O_K/\p_j : \F_p]$ is equal to the degree of $g_j$.

We now fix such a choice and change notation slightly. We simply write $\p$ for the prime ideal 
we have obtained in that way and rename the chosen $g_j$ as $g_\p$ to make the dependence on $\p$ visible.
So we have $\p = (p, g_\p(\theta))$, where $g_\p(x) \pmod p$ is an irreducible factor of $g_{\theta}(x) \pmod p$. 
Since
\[
	\Z[\theta] / (p, g_\p(\theta)) \simeq \O_K / (p, g_\p(\theta))
\]
by our choice of $p$, we abuse notation and write $\p$ for the ideal generated by $p$ and $g_\p(\theta)$
in both $\O_K$ and $\Z[\theta]$. In the same way we choose a second large prime $q\not=p$ and a prime ideal
$\q = (q, g_\q(\theta))$ above $q$. We write $f_\p$ and $f_\q$ for the residue degrees $[\O_K/\p : \F_p]$ and
$[\O_K/\q : \F_q]$, respectively.

\subsection{The public key} The public key consists of the monic irreducible polynomial $g_{\theta} \in \Z[x]$ and the ideal
$\mathfrak{n} = \p \q$. Alternatively, we may consider the polynomial $g_{\theta}$ and thus the
field $K$ as fixed. For the security of a generalised version of the Rabin cryptosystem it is crucial that it is not feasible to compute the decomposition into prime ideals of $\mathfrak{n}=\p\q$, i.e. computing $\p$ and $\q$ assuming we know $\mathfrak{n}$.

First of all we observe that this is solvable if one is able to factor the number $n=pq$, where $p$ and $q$ are the prime numbers below $\p$ and $\q$, under the assumption that we represent $\mathfrak{n}$ in a way which allows to compute $\mathfrak{n}\cap\Z=(n)$. Indeed from the factorisation of $n$, we easily compute all the primes of $\mathcal O_K$ dividing $p$ and $q$ and we just need to pick the correct ones. To decide which of the factors actually divide $\mathfrak{n}$ we can use the algorithm described in \cite[\S 4.8.3]{Cohen} to compute the valuation of $\mathfrak{n}$ with respect to the different prime ideals.

Before continuing, we need to clarify how we present $\mathfrak{n}$. Recall that $p$ is called inert in $K$ if $\p = (p)$. If both $p$ and $q$ are inert, we have $\mathfrak{n} = (n)$,
where $n = pq$,
and we can publish the generator $n$. Now suppose that at least one of $p$ and $q$ is not inert.
We claim that it is not secure to publish the four generators $n := pq$, $p g_\q(\theta)$, $q g_\p(\theta)$ and
$g_\p(\theta) g_\q(\theta)$. To see this, we may and do assume that $p$ is not inert. Then the coefficients of
$g_\p(\theta)$ are not all divisible by $p$, as otherwise $\p = (p, g_\p(\theta)) = (p)$. Hence the
greatest common divisor of the coefficients of $q g_\p(\theta)$ and $n$ is $q$ and we have factored $n$.

Let us assume that $g_\p(\theta) \in \p \setminus \p^2$ and that $g_\p(\theta)$ is coprime to $q$, and the same with the roles of $p$ and $q$ reversed.
Both conditions can be achieved in polynomial time by \cite[Proposition 1.3.11]{Cohen2}. In practice,
we can replace $g_\p(\theta)$ by $g_\p(\theta)+p$ if $g_\p(\theta) \in \p^2$. Then the first condition holds.
The probability that the second condition fails is negligibly small as $q$ has been chosen independently of
$p$ and the probability of a random $x \in \O_K$ to lie in a prime ideal above $q$ is at most $d/q$. 

Under these conditions one has $\mathfrak{n} = (n, h(\theta))$ with $n = pq$ and $h = g_\p \cdot g_\q$.
The next result shows that this is still not secure unless we choose $g_\p$ and $g_\q$ of the same degree,
i.e.\ $f_\p = f_\q$.

\begin{prop}\label{easyfactor}
Let $K = \Q(\theta)$, where $\theta$ is a root of a monic irreducible polynomial $g_{\theta}\in\Z[x]$. Let $\p$ and $\q$ be two prime ideals of $\mathcal O_K$ lying over two distinct unramified primes $p,q\in\Z$. Let  us assume that the inertia degrees $f_\p$ and $f_\q$ of $\p$ and $\q$ are distinct. Suppose that  $\mathfrak{n}=\p\q$ is given by a set of generators as described above
or in Hermite normal form, and $\mathfrak{n}\cap\Z=(n)$. Then it is possible to efficiently compute $p$ and $q$.
\end{prop}

\begin{proof}
We first show that we can factor $n$ if we can compute the norm of $\mathfrak n$, namely $\mathcal N(\mathfrak n) := [\O_K:\mathfrak n]=p^{f_\p}q^{f_\q}$. Indeed if $n^a$ is the largest power of $n$ that divides $\mathcal N(\mathfrak n)$, then $\mathcal N(\mathfrak n)/n^a$ is divisible by exactly one of the prime numbers $p$ and $q$ and $\gcd(\mathcal N(\mathfrak n)/n^a,n)$ is a prime factor of $n$.

If $\mathfrak{n}$ is described by its Hermite normal form, which is a square matrix, then the norm of $\mathfrak n$ is the determinant of the matrix by \cite[Proposition 4.7.4]{Cohen}.
So we are done in this case.

Assume now that $\mathfrak n=(n,h(\theta))$ as above. We actually need to compute $\mathcal N(\mathfrak n)$ only up to some factor which is coprime to $n$, so it is enough to consider
\[\mathcal N((h(\theta))) = \mathcal N(\p) \mathcal N(\q) \mathcal N(\mathfrak{a}) = p^{f_\p}q^{f_\q} \mathcal N(\mathfrak{a}),\]
where $\mathcal N(\mathfrak{a})$ is prime to $n$.
As the index $[\O_K: \Z[\theta]]$ is not divisible by either $p$ or $q$, it is enough to compute $[\Z[\theta]: (h(\theta))]$, which can be done for instance via computing the Smith normal form (see \cite[\S 2.4.4]{Cohen}) of
the quotient $\Z[\theta]/(h(\theta))$.
\end{proof}

By the above proposition, there is an efficient algorithm to factor $\mathfrak{n}$ if $f_\p\neq f_\q$. Since this is not desirable, we focus on the remaining case $f_\p=f_\q$.

Let us consider the special case in which $\p$ and $\q$ are inert over $\Q$, i.e. $f_\p=f_\q=[K:\Q]$. Then $\p=(p)$, $\q=(q)$ and $\mathfrak{n}=(pq)$ and factoring $\mathfrak{n}$ is clearly equivalent to factoring its generator $pq$.
Moreover, restricting to the case of inert primes has some further positive side effects: the amount of information that can be sent with our protocol becomes maximal and the computations which have to be performed are easier.

In the case $f_\p=f_\q<[K:\Q]$ we can neither provide a proof that the factorisation of $\mathfrak{n}$ is as difficult as the factorisation of $n$ nor show that there is an efficient algorithm to factor $\mathfrak{n}$.
In the case of quadratic number fields, we can offer the following result.

\begin{prop}
	Let $n=pq$ be a product of two distinct prime numbers. Suppose that there is an algorithm with the
	following property:
	For a randomly chosen non-square $\delta$ in the range $0<\delta<n$ such that the primes $p$ and $q$ split 
	in $\Q(\sqrt{\delta})$, it computes the factorisation of a randomly chosen ideal $\mathfrak{n} = \p\q$
	with probability at least $\omega > 0$,
	where $\p$ and $\q$ are prime ideals of $\mathcal O_{\Q(\sqrt{\delta})}$ dividing $p$ and $q$, respectively.
	
	Then there is an algorithm that computes the factorisation of $n$ after $k$ steps with probability at least
	$1 - (1-\omega)^k$.
\end{prop}

\begin{proof}
	Choose a random non-square $0<\delta<n$. If $\gcd(\delta,n) > 1$, we can factor $n$ so that we will
	exclude this case in the following. Let us write $\delta = b^2 \delta'$, where $\delta'$ is square-free.
	Then $\Q(\sqrt{\delta}) = \Q(\sqrt{\delta'})$ so that the index of $\Z[\delta]$ in $\mathcal O_{\Q(\sqrt{\delta})}$
	is either $b$ or $2b$ and thus coprime with $n$. As a consequence, $p$ and $q$ split in $\Q(\sqrt{\delta})$
	if and only if $\delta$ is a square modulo $p$ and $q$ and hence modulo $n$. So in this case we have
	$\delta \equiv a^2 \pmod n$ for some $0<a<n$. Note that actually $a> \sqrt{n}$ as otherwise $\delta = a^2$ would be
	a perfect square. This observation holds for all four square roots of $\delta \pmod n$ and in particular for
	$n-a$ so that we must have $\sqrt{n} < a < n-\sqrt{n}$.
	
	For the required algorithm, we first choose a random $\sqrt{n} < a < n-\sqrt{n}$ and take $0<\delta<n$
	to be a representative of the class of $a^2$ modulo $n$. If $\delta= b^2$ happens to be a square,
	then $b$ is a root of $\delta$, but $b \not\equiv \pm a \pmod n$ so that $\gcd(n,b-a)$ is either $p$ or $q$.
	If $\delta$ is a non-square, we consider the ideal $\mathfrak{n} =(n, a + \sqrt{\delta})$ in the quadratic number field
	$\Q(\sqrt{\delta})$. We have seen above that $p$ and $q$ split in $\Q(\sqrt{\delta})$. The norm
	of $a+ \sqrt{\delta}$ is $a^2-\delta = kn$ for some $0<k<n$, since $0 < \delta < n < a^2 < n^2$.
	Once more, if $k$ happens to be a multiple of either $p$ or $q$, we can factor $n$ by computing
	$\gcd(n,k)$. So let us assume that this is not the case. Then we must have
	$\mathfrak{n} = \p \q$ for prime ideals $\p$ above $p$ and $\q$ above $q$. If we can factor $\mathfrak{n}$, we find $\p$ and $\q$, and computing the gcd of their norms with $n$ we finally obtain $p$ and $q$.
	This event has probability at least $\omega$ if we can show that $\delta$ and $\mathfrak{n}$ have been chosen at random.
	
	For this, we observe that each non-square $\delta$ in the range $0<\delta<n$ such that the primes $p$ and $q$ split 
	in $\Q(\sqrt{\delta})$ is of the form $\delta \equiv a^2 \pmod n$ for exactly four choices of $a$ in the range
	$\sqrt{n} < a < n-\sqrt{n}$. If our choice $a$ leads to the ideal $\mathfrak{n} = (n,a+\sqrt{\delta}) = \p \q$,
	then $\p= (p, a+\sqrt{\delta})$ and $\q= (q, a+\sqrt{\delta})$. The second prime above $p$ is
	$\overline{\p} := (p, a-\sqrt{\delta}) = (p, n-a+\sqrt{\delta})$. We define $\overline{\q}$ and
	$\overline{\mathfrak{n}}$ similarly. Then the root $n-a$ of $\delta \pmod n$ leads to $\overline{\mathfrak{n}} = \overline{\p\q}$.
	The third root $a'$ with ${a' \equiv a \pmod p}$ and ${a' \equiv -a \pmod q}$ leads to the ideal
	$\mathfrak{n}' := (n, a'+\sqrt{\delta})$ and $\mathfrak{n}' = \p \overline{\q}$, since we have that $a'+\sqrt{\delta} \equiv a + \sqrt{\delta}
	\equiv 0 \pmod \p$, but $a'+\sqrt{\delta} \not\in \q$ as otherwise $q \mid (a-a')$.	
	Finally, the root $n-a'$ leads to $\overline{\p}\q$. So choosing $a$ at random is equivalent to
	choosing $\delta$ and $\mathfrak{n}$ at random.
	
	This finishes the proof as we can apply the above method $k$ times for different choices of $a$.
\end{proof}

\begin{remark}
Note that for a fixed $n$ and a fixed $\delta$, we are not able to find an ideal $\mathfrak{n}$ as required in the proposition, since this would be equivalent to computing a square root of $\delta$ modulo $n$, which is known to be as hard as factoring $n$.

This means that for a fixed $\delta$ and for prime ideals $\p$ and $\q$ of $\Q(\sqrt{\delta})$ dividing split prime numbers $p$ and $q$, we do not know whether $\mathfrak{n}=\p\q$ is difficult to factor, but it seems likely that this is the case. Note, however, that if $\mathfrak{n}$ is described by a set of generators, these generators might carry some information about the factorisation. 
\end{remark}

\begin{example}
Let us examine the case of quadratic fields more closely. Let $\delta\in\Z$ be square-free and let $\theta=\sqrt{\delta}$. An odd prime $p$ not dividing $\delta$ can be either inert or totally split, depending on whether the polynomial $g(x)=x^2-\delta$ is irreducible modulo $p$ or not. This condition can easily be checked by computing the Legendre symbol $\left(\frac{\delta}{p}\right)$.

Let us take for example $\delta=-5$. The primes which are inert are exactly the primes $p\equiv 11,13,17,19\pmod{20}$; so the first two primes which split completely are $p=3$ and $q=7$. As explained above (up to the sign before $\sqrt{-5}$), we get $\p=(3,1+\sqrt{-5})$ and $\q=(7,3+\sqrt{-5})$. Since $1+\sqrt{-5}$ is coprime to $7$ and $3+\sqrt{-5}$ is coprime to $3$, we obtain that
\[\p\q=(21,(1+\sqrt{-5})(3+\sqrt{-5}))=(21,-2+4\sqrt{-5}).\]
When, like in this example, both primes split completely the factorisation algorithm described in Proposition \ref{easyfactor} does not apply, but still the factorisation is at most as difficult as the factorisation of $n=pq$ in $\Z$. Since the security level remains at most the same, the involved computations are not easier to perform and the message space $\mathcal O_K/\p\q$ is isomorphic to $\Z/pq\Z$ (so, in particular, of the same size), this setting offers no advantages compared to the classical Rabin-cryptosystem and so we will not pursue this further.
\end{example}

\begin{example}
The next family of number fields we want to examine are cyclotomic fields. So let $m\in\N$ such that $m \not\equiv 2 \pmod 4$ and let $\zeta_m$ be a primitive $m$-th root of unity. It is well-known that $\Z[\zeta_m]$ is the ring of integers in $\Q(\zeta_m)$, that the primes which ramify in $\Z[\zeta_m]$ are those dividing $m$ and for any other prime number $p$, the inertia degree of $p$ is equal to its order modulo $m$. By Proposition \ref{easyfactor}, our choice of $p$ and $q$ must at least ensure that their inertia degrees coincide; actually by the considerations in the above example, we would like to take both $p$ and $q$ to be inert, which is equivalent to $p$ and $q$ having order $\varphi(m)$ modulo $m$. This can be achieved (for infinitely many primes) if and only if $m$ is a power of an odd prime or $m = 4$ (which is the case of the Gaussian integers $\Z[i]$).

Let us take for example $m=7$. Then the prime numbers which are inert in $\Z[\zeta_7]$ are those congruent to $3$ or $5$ modulo $7$, since $3$ and $5$ both have order $6$ modulo $7$.
\end{example}

\section{Square roots modulo ideals in number fields} 
Let $K$ be a number field of degree $d$ over $\Q$ with ring of integers $\O_K$ and let $\p$ and $\q$ be two distinct prime ideals of $\O_K$. Set $\mathfrak{n}=\p\q$. As before we denote the
rational primes below $\p$ and $\q$ by $p$ and $q$, respectively. 
Since for our cryptographic applications $p$ and $q$ should be large, we can assume that they are both odd. 

Note that we can compute $p$ if the ideal $\p$ is given. Indeed in our applications $\p$ will be given by
a set of generators, where the first is equal to $p$, but even if that is not the case, we can compute
the norm $\mathcal{N}(\p)$, which is a (small) power of $p$. The exponent is at most $d$.
From this, one can compute the prime $p$.
If $K$ is a Galois extension of $\Q$ (as in any of the examples discussed above), we actually know
that the exponent must be a divisor of $d$. Finally, if $p$ is inert in $K$, the exponent is actually
equal to $d$.

\subsection{Chinese Remainder Theorem revisited} For our generalisation of Rabin's cryptosystem we need an algorithm that computes the roots of any square in $\O_K/\mathfrak{n}$, under the assumption that we know the prime factors $\p$ and $\q$ of $\mathfrak{n}$. The first tool is an explicit version of the Chinese Remainder Theorem for $\O_K$ in which the involved isomorphism can be efficiently implemented on a computer.
More precisely, we need to compute the inverse of the canonical isomorphism
\begin{equation} \label{eqn:CRT}
	\O_K/\mathfrak{n} \simeq \O_K/ \p \times \O_K/ \q, \quad m+\mathfrak{n}\mapsto (m+\p,m+\q).
\end{equation}
In \cite{RabinGaussianIntegers} the authors use the extended Euclidean algorithm in $\Z[i]$ to construct an
inverse, but the ring of integers $\O_K$ is not Euclidean in general. The following algorithm solves this problem
by using only the fact that $\Z$ is a Euclidean domain.

\begin{algo}[Inverse of \eqref{eqn:CRT}]\label{crt}
Given $x + \p \in \O_K/\p$ and $y+\q \in \O_K/\q$ this algorithm finds an $m \in \O_K$ which is mapped
to the pair $(x+\p, y+\q)$ under \eqref{eqn:CRT}.
\begin{enumerate}
	\item Compute the rational primes $p$ and $q$ below $\p$ and $\q$, respectively.
	\item Use the extended Euclidean algorithm in $\Z$ to compute Bézout coefficients $a,b \in \Z$ of the equation
	$ap + bq = 1$.
	\item Output $m = xbq+yap \in \O_K$.
\end{enumerate}
\end{algo}

\subsection{Tonelli-Shanks for number rings} So we are left with the problem of computing square roots in $\O_K/\p$ and $\O_K/\q$, which can be done by a generalised version of the algorithm by Tonelli-Shanks (cf.~\cite[Note 3.41]{handbook}). We let $f$ be the inertia degree of $\p$ (the computation of $f$ can be done efficiently). By definition, the residue field $\O_K/\p$ has $p^f$ elements, so $\mathcal{N}(\p) = p^f$, and its group of units $(\O_K/\p)^\times$ has $p^f-1$ elements.

\begin{prop}[Tonelli-Shanks for number rings] \label{prop:Tonelli-Shanks}
Let $c+\p\in \O_K/\p$ be a non-zero element which is a square. Choose a non-square $m+\p\in \O_K/\p$ and $s,t\in\N$ such that $2^s\cdot t=p^f-1$ and $t$ is odd. Let $d+\p\in\O_K/\p$ be the multiplicative inverse of $c+\p\in \O_K/\p$. Then there exists $i\in\{1,2,\dots,2^{s-1}\}$ such that
\[(m+\p)^{2ti}=(d+\p)^t\]
and
\[\left((c+\p)^{(t+1)/2}(m+\p)^{ti}\right)^2=c+\p.\]
\end{prop}

\begin{proof}
Since $m+\p$ is not a square and $(\O_K/\p)^\times$ is a cyclic group of order $2^s t$, the order of $(m+\p)^t$ is exactly $2^s$. The order of $(d+\p)^t$ divides $2^{s-1}$. The existence of $i$ then follows from the fact that $(\O_K/\p)^\times$ is cyclic. 
The second formula is immediate.
\end{proof}

We note that all the numbers required in the proposition can be computed efficiently, so that Proposition \ref{prop:Tonelli-Shanks} provides an efficient way to compute a square root of $c+\p$:
\begin{enumerate}
\item A non-zero element $x+\p\in \O_K/\p$ is a square if and only if $(x+\p)^{(p^f-1)/2}=1+\p$. This condition can be efficiently checked using the Square and Multiply Algorithm and since half of the elements are non-squares, it is easy to find a non-square element by very few trials (with a high probability).
\item To compute $s$ and $t$ we simply have to divide $p^f-1$ by $2$ as many times as possible and then stop. Computing $i$ means to compute a discrete logarithm in a group of order $2^s$. This can be done efficiently by brute-force if $s$ is very small or with the Silver-Pohlig-Hellman algorithm in general.
\item To compute $d+\p$ we use the extended Euclidean algorithm in the polynomial ring $\F_p[x]$. Note that $\mathcal{O}_K/\mathfrak{p}$
is explicitly given by $\F_p[x] / (g_{\mathfrak{p}}(x))$, where $g_{\mathfrak{p}}(x) \in \F_p[x]$ is monic and irreducible.
\end{enumerate}

\subsection{Fast square root extraction} \label{subsec:fast-square-roots}
In the case of the classical Rabin scheme, i.e.\ in the case $K = \Q$, there is a more efficient way than the Tonelli-Shanks algorithm to compute square roots modulo a prime $p$ provided that $p \equiv 3 \pmod 4$.
For general number fields, we provide the following generalisation.

\begin{prop} \label{prop:fast-square-root}
If $|\O_K/\p|\equiv 3\pmod{4}$ and $c+\p\in \O_K/\p$ is a non-zero element which is a square, then
\[(c+\p)^{(|\O_K/\p|+1)/4}\]
is a square root of $c+\p$.
\end{prop}

\begin{proof}
Since $c+\p$ is a square, its order divides $(|\O_K/\p|-1)/2$. Therefore we have that
\[\left((c+\p)^{(|\O_K/\p|+1)/4}\right)^2=(c+\p)^{(|\O_K/\p|+1)/2}=(c+\p)^{(|\O_K/\p|-1)/2}(c+\p)=c+\p,\]
as desired.
\end{proof}

By the discussion in the previous section, we are mainly interested in primes which are inert and by the above proposition we have a gain in efficiency in the computations if we restrict to the case $|\O_K/\p|\equiv 3\pmod{4}$. Unfortunately, these two conditions are incompatible in quadratic number fields, since then for an inert odd prime $p$, $|\O_K/\p|=p^2\equiv 1\pmod{4}$. This motivates us to examine cubic number fields as $p^3 \equiv p \pmod 4$. Since we also want an efficient way of checking if a prime is inert, it turns out that taking fields of the form $\Q(\sqrt[3]{d})$ 
might not be the optimal choice.

\begin{example}
Some very interesting examples of cubic extensions with good characterisations for families of inert primes are subfields of cyclotomic extensions.

The first example we consider is the maximal real subfield of $\Q(\zeta_7)$, i.e. the field $\Q(\zeta_7+\overline{\zeta_7})$. 
By \cite[Proposition 2.16]{Washington} the ring of integers of that field is $\Z[\zeta_7+\overline{\zeta_7}]$. We can easily compute the minimal polynomial of $\zeta_7+\overline{\zeta_7}$, namely $x^3+x^2-2x-1$. The primes which are inert in $\Q(\zeta_7+\overline{\zeta_7})$ are those with inertia degree divisible by $3$ in $\Q(\zeta_7)$. We have already seen that 
these are the primes whose classes modulo $7$ have an order which is a multiple of $3$, namely those congruent to $2$, $3$, $4$ or $5$ modulo $7$. In order to use the simplified computation of square roots above, we restrict to primes congruent to $3$ modulo $4$. Summarizing, we want to take prime numbers congruent to $3$, $11$, $19$ or $23$ modulo $28$.

Similarly, we can consider the maximal real subfield of $\Q(\zeta_9)$, i.e. $\Q(\zeta_9+\overline{\zeta_9})$. The minimal polynomial of $\zeta_9+\overline{\zeta_9}$ is $x^3-3x+1$ and a family of inert primes consists of those which are congruent to $2$, $4$, $5$ or $7$ modulo $9$. This leads to considering the primes congruent to $7$, $11$, $23$ or $31$ modulo $36$.

This can be done for every cyclotomic field whose degree is a multiple of $3$. From $\Q(\zeta_{13})$ we obtain the extension given by a root of $x^3+x^2-4x+1$; from $\Q(\zeta_{19})$ we get $x^3+x^2-6x-7$. Note that in this last example, the largest power of $3$ dividing the degree of the cyclotomic field over $\Q$ is $9$, so in order to make sure that a prime is inert in the subfield of degree $3$ we need to require that the order of $p$ modulo $19$ is a multiple of $9$, not just of $3$, which leads to the primes congruent to $2$, $3$, $4$, $5$, $6$, $9$, $10$, $13$, $14$, $15$, $16$, $17$ modulo $19$. Together with the usual requirement ${p\equiv 3\pmod{4}}$, we get the acceptable congruences to $3$, $15$, $23$, $35$, $43$, $47$, $51$, $55$, $59$, $63$, $67$, $71$ modulo $76$.
\end{example}

\subsection{Square roots and factoring} We end this section with an examination of the converse problem, which will be relevant in the discussion about the security of our cryptosystem.

\begin{prop}\label{roots2factor}
	Let $K$ be a number field and let $\mathfrak n = \mathfrak{p}\mathfrak{q}$ be the product of two distinct
	prime ideals in $\mathcal{O}_K$. Suppose $m_1+\mathfrak{n}$ and $m_2+\mathfrak{n}$ 
	are both square roots of the same element
	$c+\mathfrak{n}\in\O_K/\mathfrak{n}$ and that $m_1+\mathfrak{n}\neq \pm m_2+\mathfrak{n}$. 
	Then we can efficiently compute a factorisation of $\mathfrak{n}$.
\end{prop}

\begin{proof}
	By the Chinese Remainder Theorem \eqref{eqn:CRT}, $m_1$ and $m_2$ are congruent modulo one prime divisor of $\mathfrak{n}$ 
	and are negatives of each other modulo the other.
	Hence $a := m_1+m_2$ belongs to exactly one of the prime divisors of $\mathfrak{n}$, say $\p$.
	Then the ideal $\mathfrak{n} + (a)$ is strictly larger than $\mathfrak{n}$ and contained in $\p$.
	This forces $\p = \mathfrak{n} + (a)$ and by a similar reasoning one has $\mathfrak{q} = \mathfrak{n} + (b)$,
	where $b := m_1 - m_2$.	
\end{proof}

\section{Choosing the right square root}
An important issue of the Rabin-cryptosystem, even in its classical version over $\Z$, is the non-uniqueness in the decryption process: While encrypting corresponds to squaring, decrypting means to compute a square root modulo $n=pq$, and in general there are four of them. This question has been addressed in \cite{EliaPivaSchipani}, where the authors describe several ways of transmitting the required information in two extra bits. The easiest method consists in taking as extra bits the parity of the message $m$ and its Jacobi symbol over $n$. 
This system, which was originally proposed in \cite{Williams}, is shown to work when both $p$ and $q$ are congruent to $3$ modulo $4$.
Indeed the two square roots having the same parity then have opposite Jacobi symbols as explained in \cite[\S 3.1]{EliaPivaSchipani}.
Note that this congruence condition is also a good choice in order to simplify the computation of the square roots. Of course, using some paddings with random bits, one has to make sure that the two extra bits do not leak any information about the message, which would compromise semantic security.

We will now see how to handle the issue of identifying the right square root in the generalised setting treated in this paper. The requirement of taking inert primes $p$ and $q$ turns out to be very helpful also at this stage, so we will keep assuming it. Assume that the message we encrypt is of the form $m=\sum_{i=0}^{d-1}a_i\theta^i$, where $a_i \in \Z/n\Z$
for $0 \leq i < d$. Then the decryption process leads to four possible square roots $m_1,m_2,m_3,m_4$ of $c= m^2$. Write $m_j=\sum_{i=0}^{d-1} a_{i,j}\theta^i$ with $a_{i,j} \in \Z/ n\Z$. Then each coefficient $a_{i,j}$ must be congruent to $\pm a_i$ modulo $p$ and modulo $q$, possibly with opposite signs. In any case $a_{i,1},a_{i,2},a_{i,3},a_{i,4}$ are exactly the four square roots of $a_i^2$ modulo $n$. Therefore to determine which of the possible messages $m_1,m_2,m_3,m_4$ is the original message $m$ is equivalent to determining for a fixed $i$ which of the coefficients $a_{i,1},a_{i,2},a_{i,3},a_{i,4}$ coincides with $a_i$. This means that we can literally use the same methods discussed in \cite{EliaPivaSchipani} and apply them to a single component of the message, for example the one of $\theta^0$.

There is still one detail which needs to be addressed: some of the $a_{0,1},a_{0,2},a_{0,3},a_{0,4}$ might coincide, which happens if and only if any one of the numbers $a_{0,j}$ is congruent to $0$ modulo $p$ or modulo $q$. To avoid this, we require that the coefficient $a_0$ 
(or more generally the first non-zero coefficient) of the plaintext message must satisfy $\gcd(a_0,n)=1$. This excludes the case $a_0=0$ and the cases where $\gcd(a_0,n)=p$ or $q$. Since $n$ is assumed to be hard to factor and any $a_0$ satisfying $\gcd(a_0,n)=p$ or $q$ would yield a factorisation of $n$, such an $a_0$ is extremely unlikely to be chosen. Thus the restriction on the message space is of no practical relevance.


\section{The Rabin cryptosystem in number fields}
We can now introduce the specifications of our generalised Rabin cryptosystem to number fields. 

\subsection{Key generation} First of all we need to choose a number field. This choice could be part of the public key, but there is actually no reason for using a different field every time; indeed, fixing a field $K$ once and for all permits to increase the efficiency of key generation by doing some preliminary computations.

We will distinguish and compare three different implementations for the following families of number fields: quadratic fields, (subfields of) cyclotomic fields, general number fields. Since for the key generation algorithm, we want to choose prime numbers that are inert in $K$, it is useful to find some congruence conditions which are easy to verify. Let us see how to do this in the three mentioned cases.\\

Case 1. A quadratic field $K=\Q(\sqrt{\delta})$, where $\delta$ is a square-free integer, which is not too big. A 
sufficiently large prime $p$ will not ramify in $K$ and it is inert if and only if $x^2-\delta$ is irreducible modulo $p$, i.e. if and only if $\left(\frac{\delta}{p}\right)=-1$. Using quadratic reciprocity and its supplements, we can find a list of congruences modulo $\delta$, $4\delta$ or $8\delta$ which $p$ must satisfy.

Case 2. A cyclotomic field $K=\Q(\zeta_m)$, where $m = \ell^k$ is a perfect power of an odd prime number $\ell$. Then a prime $p$ is inert if and only if its class modulo $m$ is a primitive element modulo $m$. By computing one primitive element modulo $m$ and taking all of its powers with exponents coprime to $\varphi(m) = (\ell-1) \ell^{k-1}$, we find all the possible primitive elements and, as in case 1, we obtain a list of possible congruences modulo $m$ which $p$ has to satisfy.

If $\ell-1 = 3^a b$ with $a>0$ and $3 \nmid b$, then there is a unique subfield $K_3$ of $K$ of degree $3$ over $\Q$. 
A prime $p$ is inert in $K_3$ if and only if its inertia degree in $K$ is divisible by $3^a$. Again this leads
to a list of possible congruences modulo $m$.

Case 3. A generic number field $K$ is defined by the minimal polynomial $g$ of a primitive element $\theta$ of $K$. Then the prime numbers $p$ which are inert are characterised as those primes modulo which $g$ remains irreducible. This condition needs to be checked prime by prime in the key generation algorithm, making the computation less efficient. It might indeed happen that for a given number field we do not find any suitable primes.\\

Let us describe the key generation algorithm for a fixed field $K$. We need to choose two prime numbers $p$ and $q$ satisfying the conditions described above and such that their product is not factorisable efficiently with any known algorithm. In particular we need to make sure that $p$ and $q$ are sufficiently large and are such that $p-1$ and $q-1$ have some large primes among their factors. Here is a good way of achieving this if we have some congruence conditions modulo some integer $D$ like in case 1 and case 2: choose a random value $c$ among the allowed congruences, choose random numbers $k$ larger than $2^{\lambda}$, where $\lambda$ is a fixed security parameter, until $\ell=2kD+1$ is prime. Finally take $h=0$ and increase it by $1$ until $p=(hD+c-1)\ell+1$ is prime. By construction $p$ is a large prime congruent to $c$ modulo $D$ and such that $p-1$ is divisible by the large prime $\ell$. The same construction is used to choose $q$. The pair $(p,q)$ is the private key, the product $n=pq$ is the public key. In the generic case 3 we can just choose random primes $p$ and $q$ and need to check prime by prime whether the polynomial $g$ factorises modulo $p$ and $q$; this can be quite time consuming.

In the cases where we cannot apply the fast square root extraction of \S \ref{subsec:fast-square-roots}, 
the two non-squares which are required for the Tonelli-Shanks algorithm can be determined once
and then stored for further use as part of the private key. This speeds up decryption slightly.

\subsection{Rabin scheme over number rings}
Suppose we have given a monic irreducible polynomial $g \in \Z[x]$ of degree $d$ and a public key $n = pq$, where $p$ and $q$
are two distinct primes, which are inert in $K = \Q[x] / (g)$ and congruent to $3$ modulo $4$. 
We let $\theta := x \pmod{g}$ so that $\Z[\theta] = \Z[x]/(g)$.
If $0 \not= m = \sum_{i=0}^{d-1} a_i \theta^i \in \Z[\theta] / (n)$ with $a_i \in \Z/ n\Z$, $0 \leq i <d$, then we write
$\alpha(m)$ for the first non-zero coefficient of $m$, i.e.\
$\alpha(m) := a_{i_0}$, where $a_{i_0} \not=0$ and $a_i = 0$ for all $i<i_0$. 

Here and in the following, for each class in $\mathbb{Z}/n\mathbb{Z}$ we use its representative in
$\{0,\ldots,n-1\}$. 
Recall that if $\gcd(\alpha(m),n) \not=1$, it reveals the factorisation of $n$. Such a message will be rejected by the encryption
algorithm, but should never occur in practice.
In our implementation we will assume for simplicity that $a_0\neq 0$.

\subsection*{Encryption} Suppose Alice wants to encrypt a (non-zero) message 
$m = \sum_{i=0}^{d-1} a_i \theta^i$ with
$a_i \in \Z/ n\Z$, $0 \leq i <d$. If $\gcd(\alpha(m), n) \not=1$ reject.
Otherwise she computes $c \equiv m^2 \pmod{n}$, the parity $b_0  = \alpha(m) \pmod 2$ 
of $\alpha(m)$ and $b_1 = \frac{1}{2}\left(1- \left(\frac{\alpha(m)}{n}\right)\right) \in \left\{0,1\right\}$,
where $\left(\frac{\cdot}{n}\right)$ denotes the Jacobi symbol over $n$, and transmits
the triple $(c, b_0,b_1)$ to Bob.

\subsection*{Decryption} Bob computes the square roots of $c \pmod{p}$ and $c\pmod{q}$ either with the Tonelli-Shanks
algorithm (Proposition \ref{prop:Tonelli-Shanks}) or - whenever possible -  via Proposition \ref{prop:fast-square-root}.
He computes the Legendre symbols modulo $p$ and modulo $q$ of the first non-zero coefficients of these square roots.
Via Algorithm \ref{crt} he combines two of the latter to a square root $m_1$ of $c\pmod n$ having the correct value of $b_1$. Note that $-m_1$ then leads to the same value of $b_1$ and is also a square root of $c \pmod n$. The original message is the one among $m_1$ and $-m_1$ having the correct parity bit $b_0$.\\

\begin{thm}
In the setting of this section, if an attacker can decrypt a uniformly random admissible message $(c,b_0,b_1)$ 
with a non-negligible probability $\omega$ in polynomial time, 
then he can factor $\mathfrak{n}$ in polynomial time with probability at least $\omega$ in one trial.
After $k$ independent trials, the success probability is at least $1 - (1-\omega)^k$.
In particular, if the attacker can extract a square root of every square modulo $\mathfrak{n}$,
one can take $\omega = 0.5$.
Conversely, knowledge of the factorisation $\mathfrak{n} = \p\q$ permits the efficient
decryption of every admissible ciphertext.
\end{thm}

\begin{proof}
The attacker starts picking a random message $m$ and encrypts it to $(c,b_0,b_1)$. With a probability $\omega$ he will be able to decrypt the ciphertext $(c,b_0,1-b_1)$, which is a valid ciphertext, whose decryption is a message $\tilde m$ such that $m+\mathfrak{n}\neq \pm \tilde m+\mathfrak{n}$ and $m^2+\mathfrak{n}=\tilde m^2+\mathfrak{n}$. By Proposition \ref{roots2factor} this is enough to efficiently compute a factorisation of $\mathfrak{n}$. Note that the ciphertext $(c,b_0,1-b_1)$ is also uniformly random as changing $b_1$ is
a bijection on the set of admissible ciphertexts.

Repeating the above steps $k$ times for different choices of $m$, the probability of succeeding is $1 - (1-\omega)^k$.
The remaining statements are clear from the above considerations.
\end{proof}

\begin{remark}
	The preceding theorem concerns only the computational relation between factorization,
	square root extraction, and recovery of the plaintext from the complete ciphertext. It does not assert semantic security of the resulting encryption scheme. In fact, the encryption algorithm described above is deterministic and therefore cannot be semantically secure. Semantic security would require some additional randomized padding. The specification and security analysis of such a procedure are beyond the scope of the present article.
\end{remark}

We provide an implementation of the described algorithms for the settings described above in PARI on GitHub. We also included an implementation of the classical Rabin cryptosystem in order to guarantee a fair comparison of running times.\footnote{https://github.com/alecobbe/RabinNF}

\section{Security and efficiency}

In this section we examine some aspects of the implementation in more detail and compare running times of our algorithms in different cases.

For our runtime experiments, we ran our implementations using PARI/GP Calculator version 2.17.3 (64-bit, single-threaded engine). All computations were performed on a machine equipped with an AMD Ryzen 7 5825U processor (2.0 GHz base frequency) and 16 GB of RAM, running Windows 11 Home. Since the computer measures time in milliseconds and the computations are quite fast, we did every encryption 10,000 times and every decryption $100$ times in the classical setting and $10$ times over number fields in order to get meaningful results. We computed 1,000 encryptions and decryptions of random messages and with random keys in the classical setting, over the Gaussian integers and the cubic ring $\Z[\zeta_7+\overline{\zeta_7}]$ defined by the polynomial $x^3+x^2-2x-1$. The prime numbers had key sizes of about $2^{2048}$, so that the public keys had sizes of about $2^{4096}$. In order to improve comparability, we took exactly the same 1,000 keys in each of the three settings (the congruence conditions are compatible to each other). We observed that the majority of the measured execution times were closely clustered, with only a few sporadic outliers. Consequently, we used the median instead of the arithmetic average to report the execution times.

The encryption function is literally the same in all cases, since it just has to compute a square in the relevant ring of integers modulo the public key. Note, however, that the running time increases with the degree of the number field, since we need to compute a square in a larger ring, whose elements are represented by polynomials. Our results are summarized in the following table.

\begin{table}[htbp]
	\centering
	\caption{Comparison of running times for encryption}
	\begin{tabular}{lccc}
		& $\mathbb Z$
		& $\mathbb Z[i]$
		& $\mathbb Z[\zeta_7+\overline{\zeta_7}]$ \\
		\hline
		Running time (ms)
		& 0.1526 & 0.1651 & 0.1955 \\
		
		Relative running time
		& 1 & 1.082 & 1.281 \\
		
		Relative running time per message size
		& 1 & 0.5410 & 0.4270 \\
	\end{tabular}
\end{table}

In the first line of the table we list the running times for encryption in milliseconds, in the second line we included the factors by which the computations are slower compared to $\Z$, while in the last line we normalize the factors with respect to the plaintext size. The results are actually very favourable.

Concerning decryption there are two different scenarios: for $\Z[i]$ we need to use the Tonelli-Shanks algorithm, while for $\Z$ and $\Z[\zeta_7+\overline{\zeta_7}]$ we can use the simplified version of Proposition \ref{prop:fast-square-root}. We summarize the results in the following table.

\begin{table}[htbp]
	\centering
	\caption{Comparison of running times for decryption}
	\begin{tabular}{lccc}
		& $\mathbb Z$
		& $\mathbb Z[i]$
		& $\mathbb Z[\zeta_7+\overline{\zeta_7}]$ \\
		\hline
		Running time (ms)
		& 8.260 & 177.4 & 231.4 \\
		
		Relative running time
		& 1 & 21.48 & 28.01 \\
		
		Relative running time per message size
		& 1 & 10.74 & 9.338 \\
	\end{tabular}
\end{table}

Over the Gaussian integers, almost the whole time is required for the Tonelli-Shanks algorithm, which was our motivation to consider fields of degree $3$, where a more efficient algorithm is available. Taking the ratio with the number of bits of the message this is, as expected, better than over Gaussian integers, but it is still less efficient than the classical Rabin cryptosystem.

Let us explain what happens from a more theoretical point of view. Computing the $e$-th power of a number modulo $p$ with the square and multiply algorithm means computing $c\log(e)$ multiplications within $\Z/p\Z$ for some constant $c$. If $\mathcal O_K$ is the ring of integers of a number field of degree $3$ over $\Q$ and $p$ is a prime which is inert in $K$ and $p \equiv 3 \pmod 4$, then the computation of an $e$-th power in $\mathcal O_K/(p)$ corresponds to $9c\log(e)$ multiplications in $\Z/p\Z$. However, the powers we need to compute
over $\Z$ and $\mathcal{O}_K$ have exponents $\frac{p+1}{4}$ and $\frac{p^3+1}{4}$, respectively. Having a third power in the second exponent implies another factor of $3$ in the required number of multiplications in $\Z/p\Z$. To conclude, this analysis would explain a factor of $27$; the actually observed factor is just slightly larger.
Applying this argument formally in degree $2$, it would predict a factor $2^3=8$. This estimate is not applicable to the Gaussian case, however, because Proposition \ref{prop:fast-square-root} cannot be used. The implementation instead performs the full Tonelli-Shanks algorithm, which accounts for the substantially larger observed factor $21.48$.

To conclude, key generation takes about the same time for the classical Rabin cryptosystem and its generalisations, provided we have a list of congruences that guarantee that a prime $p$ is inert in the considered field. The encryption is extremely fast and becomes even more efficient in number fields, while decryption becomes less efficient in general number fields, even though some extensions of degree $3$ of the rational numbers perform better than Gaussian integers.

As for security, we should say that our implementation was mainly conceived to examine running times, but it uses the random function of PARI/GP, which is not cryptographically secure. Therefore our implementation should be considered an experimental prototype rather than a production-ready cryptographic implementation. Besides that issue, we have already discussed that being able to decrypt any ciphertext is equivalent to being able to factorise the public key into its prime factors, exactly as for the classical Rabin cryptosystem. In particular, it cannot be used for post-quantum cryptography since
an attacker having access to a quantum computer can factor the modulus via Shor's algorithm \cite{Shor}.
To achieve semantic security, one should also introduce some randomized padding, just as for every public key scheme. We additionally point out that we are transmitting two extra bits carrying information about the degree zero term of the polynomial representing our message; so that is where the padding should happen.

\section{Conclusion}
In this article we proposed a novel public-key cryptosystem which extends the classical Rabin scheme to
general number fields. 
We have analysed both its computational performance and its security, with particular emphasis on the influence of the arithmetic properties of the prime ideals dividing the modulus.

Encryption in the resulting scheme performs very fast, whereas decryption is generally more costly than in the classical Rabin cryptosystem. However, the more sophisticated variant developed for number fields of degree $3$ performs comparatively better than the corresponding construction for quadratic number fields. Since public-key cryptosystems are typically used only to exchange keys for more efficient symmetric encryption schemes, this loss in performance need not constitute a serious practical drawback.


On the security side, we have shown that the factorisation problem underlying our scheme is provably equivalent to the factorisation of the integer $n$ when the primes dividing $n$ are inert in the number field $K$. The situation changes substantially when the two prime ideals
$\mathfrak{p}$ and $\mathfrak{q}$ dividing the modulus $\mathfrak{n}$
 have different inertia degrees: in this case, we have shown that the factorisation problem can be solved efficiently. It remains open whether factorisation is computationally infeasible when the inertia degrees of $\mathfrak p$ and $\mathfrak q$ coincide but are strictly smaller than the degree of the number field. In any case, inert primes provide the most natural choice for our construction. Besides giving the strongest security guarantee established here, they have additional desirable properties, most notably that the resulting message space has maximal size. The same considerations also apply to the corresponding generalisation of RSA to rings of integers in number fields.

From a practical point of view, our results do not suggest that passing to number fields leads to an improvement over the classical Rabin scheme. Although the degree-$3$ construction performs better, relative to the amount of information transmitted, than the corresponding construction over the Gaussian integers, this advantage disappears when it is compared with the classical scheme over $\mathbb Z$. The principal interest of the number field setting is therefore theoretical: it provides a natural generalisation of the Rabin cryptosystem and reveals how arithmetic invariants of the underlying number field, in particular the splitting behaviour of the primes involved, govern both the security and the efficiency of the resulting cryptosystem.

%

\bibliography{bibliography}
\addcontentsline{toc}{section}{Bibliography}
\bibliographystyle{abbrv}

\end{document}